\documentclass[12pt]{article}

\usepackage{amsmath}
\usepackage{amssymb}
\usepackage{amsthm}
\usepackage{marginnote}
\usepackage{graphicx}
\usepackage{natbib}
\usepackage[top=1in, bottom=1in, left=1in, right=1in]{geometry}

\ifdefined \theorem
\else
\theoremstyle{definition}
\newtheorem{theorem}{Theorem}

\newtheorem{proposition}[theorem]{Proposition}

\newtheorem{definition}[theorem]{Definition}

\fi

\newcommand{\ip}[1]{\langle #1 \rangle}

\newcommand{\mycomment}[1]{}

\newcommand{\var}{\text{var}}

\newcommand{\tr}{\text{ tr }}

\newcommand{\bbE}{\mathbb{E}}

\newcommand{\bbN}{\mathbb{N}}

\newcommand{\bbR}{\mathbb{R}}

\newcommand{\mcD}{\mathcal{D}}

\newcommand{\mcF}{\mathcal{F}}

\newcommand{\mcU}{\mathcal{U}}

\newcommand{\ra}{\rightarrow}

\newcounter{parnum}
\newcommand{\npoint}{%
  \noindent\refstepcounter{parnum}%
  \makebox[0.5in][c]{\textbf{\arabic{parnum}.}} %
  \marginnote{\small\ttfamily\the\inputlineno}}

\newcommand{\chol}{\text{upper chol}}

\newcommand{\plim}[1]{\underset{#1}{\operatorname{plim}}\;}
\newcommand{\sym}{\mathcal{S}}

\newcommand{\bA}{\mathbf{A}}
\newcommand{\bB}{\mathbf{B}}
\newcommand{\bC}{\mathbf{C}}
\newcommand{\bD}{\mathbf{D}}
\newcommand{\bH}{\mathbf{H}}
\newcommand{\bI}{\mathbf{I}}
\newcommand{\bL}{\mathbf{L}}
\newcommand{\bp}{\mathbf{p}}

\newcommand{\br}{\mathbf{r}}
\newcommand{\bS}{\mathbf{S}}
\newcommand{\bT}{\mathbf{T}}
\newcommand{\bU}{\mathbf{U}}
\newcommand{\bV}{\mathbf{V}}
\newcommand{\bw}{\mathbf{w}}

\newcommand{\bX}{\mathbf{X}}
\newcommand{\bY}{\mathbf{Y}}
\newcommand{\bZ}{\mathbf{Z}}
\newcommand{\bSig}{\mathbf{\Sigma}}
\newcommand{\bSigma}{\mathbf{\Sigma}}
\newcommand{\bPsi}{\mathbf{\Psi}}
\newcommand{\bmu}{\mathbf{\mu}}

\begin{document}

\title{A Tractable State-Space Model for Symmetric Positive-Definite Matrices}

\date{Latest: \today, Original: October 8, 2013}

\author{Jesse Windle and Carlos M. Carvalho}

\maketitle

\begin{abstract}
  Bayesian analysis of state-space models includes computing the posterior
  distribution of the system's parameters as well as filtering, smoothing, and
  predicting the system's latent states.  When the latent states wander around
  $\bbR^n$ there are several well-known modeling components and computational
  tools that may be profitably combined to achieve these tasks.  However, there
  are scenarios, like tracking an object in a video or tracking a covariance
  matrix of financial assets returns, when the latent states are restricted to a
  curve within $\bbR^n$ and these models and tools do not immediately apply.
  Within this constrained setting, most work has focused on filtering and less
  attention has been paid to the other aspects of Bayesian state-space
  inference, which tend to be more challenging.  To that end, we present a
  state-space model whose latent states take values on the manifold of symmetric
  positive-definite matrices and for which one may easily compute the posterior
  distribution of the latent states and the system's parameters, in addition to
  filtered distributions and one-step ahead predictions.  Deploying the model
  within the context of finance, we show how one can use realized covariance
  matrices as data to predict latent time-varying covariance matrices.  This
  approach out-performs factor stochastic volatility.

  \vspace{12pt}
\end{abstract}

\tableofcontents


\section{Introduction}
\label{sec:intro}

A state-space model is often characterized by an observation density $f(y_t |
x_t)$ for the responses $\{y_t\}_{t=1}^T$ and a transition density $g(x_t |
x_{t-1})$ for the latent states $\{x_t\}_{t=1}^T$.  Usually, the latent states
can take on any value in $\bbR^n$; however, there are times when the states or
the responses are constrained to a manifold embedded in $\bbR^n$.  For instance,
econometricians and statisticians have devised symmetric positive-definite
matrix-valued statistics that can be interpreted as noisy observations of the
conditional covariance matrix of a vector of daily asset returns.  In that case,
it is reasonable to consider a state-space model that has covariance
matrix-valued responses (the statistics) and covariance matrix-valued latent
quantities (the time-varying covariance matrices).

Unfortunately, devising state-space models on curved spaces (like the set of
covariance matrices) that lend themselves to Bayesian analysis is not easy.
Just writing down the observation and transition densities can be difficult in
this setting, since one must define distributions on curved spaces.  Asking that
these densities then lead to some recognizable posterior distribution for the
latent states and the system's parameters compounds the problem.  Filtering is
slightly less daunting, since one can appeal to sequential methods.  (Filtering
or forward filtering refers to iteratively deriving the filtered distributions
$p(x_t | \mcD_t, \theta)$ where $\mcD_t$ is the data $\{y_s\}_{s=1}^t$ and
$\theta$ represents the system's parameters.)  Approximate methods can also
flounder.  For instance, when $y_t$ and $x_t$ exist in planar spaces, it is
common to write down observation and evolution equations and only specify the
first and second moments of the evolution innovations and observation errors,
which leads to tractable methods for filtering.  However, the notions of mean
and variance do not translate automatically to curved spaces and hence even this
density-less approach runs into trouble.

Despite these difficulties, it is still of interest to develop state-space
models for the set of covariance matrices and other curved spaces, since such
data exists.  In addition to the financial application described previously,
time-varying covariance matrices arise in computer vision
\citep{porikli-etal-2006}, and time varying linear subspaces arise in subspace
tracking \citep{srivastava-klassen-2004}.  Some work has explored forward
filtering such data.  \cite{tyagi-davis-2008} develop a Kalman-like filter
\citep{kalman-1960} for symmetric positive-definite matrices while
\cite{hauberg-etal-2013} develop an algorithm similar to the unscented Kalman
filter \citep{julier-uhlmann-1997} for geodesically complete manifolds.  Several
collaborations have made use of particle filters: \cite{srivastava-klassen-2004}
for the Grassmann manifold, \cite{tompkins-wolfe-2007} for the Steifel manifold,
\cite{kwon-park-2010} for the affine group, and \cite{choi-christensen-2011} for
the special Euclidean group.

None of these approaches have produced a state-space model amenable to fully
Bayesian inference---inference in which one can compute both the conditional and
marginalized versions of the filtered, smoothed, and predictive distributions
(where the conditioning and marginalizing is with respect to $\theta$) in
addition to the joint posterior distribution of the latent states and the
system's parameters.  (Smoothing refers to computing the distribution of past
states while predicting refers to computing the distribution of future states.)
One attempt in that direction, \cite{prado-west-2010} (p.\ 273), partially
address these issues for dynamic covariance matrices, but informally and with
less flexibility than the forthcoming.

We fully address these issues for a state-space model with symmetric
positive-definite or positive semi-definite rank-$k$ observations and symmetric
positive-definite latent states.  (Let $\sym_{m,k}^+$ denote the set of order
$m$, rank $k$, symmetric positive semi-definite matrices and let $\sym_m^+$
denote the set of order $m$, symmetric positive-definite matrices.)  The model
builds on the work of \cite{uhlig-1997}, who showed how to construct a
state-space model with $\sym_{m,1}^+$ observations and $\sym_{m}^+$ hidden
states and how, using this model, one can forward filter in closed form.  We
extend his approach to observations of arbitrary rank and show how to forward
filter, how to backward sample, and how to marginalize the hidden states to
estimate the system's parameters, all without appealing to fanciful MCMC
schemes.  (Backward sampling refers to taking a joint sample of the posterior
distribution of the latent states $p(\{x_t\}_{t=1}^T | \mcD_T, \theta)$ using
the conditional distributions $p(x_t | x_{t+1}, \mcD_t, \theta)$.)  The model's
estimates and one-step ahead predictions are exponentially weighted moving
averages (also known as geometrically weighted moving averages).  Exponentially
weighted moving averages are known to provide simple and robust estimates and
forecasts in many settings \citep{brown-book-1959}.

\subsection{A comment on the original motivation}

Our interest in covariance-valued state-space models arose from studying the
realized covariance statistic, which within the context of finance, roughly
speaking, can be thought of as a good estimate of the covariance matrix of a
collection of daily asset returns.  (The daily period is somewhat arbitrary; one
may pick any reasonably ``large'' period.)  We had been exploring the
performance of factor stochastic volatility models, along the lines of
\cite{aguilar-west-2000}, which use daily returns, versus exponentially weighted
moving averages of realized covariance matrices and found that exponentially
smoothing realized covariance matrices out-performed the more complicated factor
stochastic volatility models.  (Exponential smoothing refers to iteratively
calculating a geometrically weighted average of observations and some
initialization parameter.)  As Bayesians, we wanted to find a model-based
approach that is capable of producing similar results and the following fits
within that role.  To that end, as shown in Section \ref{sec:example}, this
simple model, used in conjunction with realized covariances, provides better
one-step ahead predictions of daily covariance matrices than factor stochastic
volatility (which only uses daily returns).

However, the specificity of this original application distracts from the larger
problem of devising state-space models on the set of covariance matrices or on
curved spaces more generally.  As noted previously, it is difficult to devise
tractable state-space models in this setting.  It is within this more general
problem that we find the model most notable.

\section{A Covariance Matrix-Valued State-Space Model}
\label{sec:the-model}

The model herein is closely related to several models found in the Bayesian
literature, all of which have their origin in variance discounting techniques
\citep{quintana-west-1987, west-harrison-book-1997}. \cite{uhlig-1997} provided
a rigorous justification for variance discounting, showing that it is a form of
Bayesian filtering for covariance matrices, and our model can be seen as a
direct extension of Uhlig's work.  (\cite{shephard-1994} constructs a similar
model, but only for the univariate case.)  The model of \cite{prado-west-2010}
(p.\ 273) is similar to ours, though less flexible.

\cite{uhlig-1997} considers observations, $\br_t \in \bbR^m$, $t=1, \ldots, T$,
that are conditionally normal given the hidden states $\{\bX_t\}_{t=1}^T$, which
take values in $\sym^+_m$.  We will henceforth write vectors in bold lower case
and matrices in bold upper case.  In particular, assuming $\bbE[\br_t] = 0$, his
model is
\begin{displaymath}
\begin{cases}
\br_t \sim N(0, \bX_{t}^{-1}), \\
\bX_t = \bT_{t-1}' \bPsi_t \bT_{t-1} / \lambda, & \bPsi_t \sim \beta_m \Big(
\frac{n}{2}, \frac{1}{2} \Big), \\
\bT_{t-1} = \chol \; \bX_{t-1} ,
\end{cases}
\end{displaymath}
where $n > m-1$ is an integer and $\beta_m$ is the multivariate beta
distribution, which is defined in Section \ref{sec:proofs}.  This model
possesses closed form formulas for forward filtering that only requires knowing
the outer product $\br_t \br_t'$; thus, one may arrive at equivalent estimates
of the latent states by letting $\bY_t = \br_t \br_t'$ and using the observation
distribution
\[
\bY_t \sim W_m(1, \bX_t^{-1})
\]
where $W_m(1,\bX_t^{-1})$ is the order $m$ Wishart distribution with $1$ degree
of freedom and scale matrix $\bX_t^{-1}$ as defined in Section \ref{sec:proofs}.
We show that one can extend this model for $\bY_t$ of any rank:
\begin{equation}
\label{eqn:the-model}
\tag{UE}
\begin{cases}
\bY_t \sim W_m(k, (k \bX_t)^{-1}), \\
\bX_t \sim \bT_t' \bPsi_t \bT_{t-1} / \lambda, & \bPsi_t \sim \beta_m \Big( \frac{n}{2},
\frac{k}{2} \Big), \\
\bT_{t-1} = \chol \; \bX_{t-1},
\end{cases}
\end{equation}
where $n > m-1$ and $k$ is an integer less than $m$ or is a real number greater
than $m-1$.  (When $k$ is an integer less than $m$, $\bY_t$ has rank $k$.)  Many
of the mathematical ideas needed to motivate Model \ref{eqn:the-model} (for
Uhlig extension) can be found in a sister paper \citep{uhlig-1994} to the
\cite{uhlig-1997} paper, and Uhlig could have written down the above model given
those results; though, he was focused specifically on the rank-deficient case,
and the rank-1 case in particular, as his 1997 work shows.  We contribute to
this discourse by constructing the model in a fashion that makes sense for
observations of all ranks, show that one may backward sample to generate a joint
draw of the hidden states, and demonstrate that one may marginalize the hidden
states to estimate the system's parameters $n$, $k$, and $\lambda$.

Model \ref{eqn:the-model} has a slightly different form and significantly more
flexibility than the model of \cite{prado-west-2010} (see p. 273), which is
essentially
\[
\begin{cases}
  \bY_t \sim W_m(\eta, \bX^{-1}_{t-1}), & \eta \in \{1, \ldots, m-1\} \cup (m-1,\infty), \\
  \bX_{t} = \bT_{t-1}' \bPsi_t \bT_{t-1} / \lambda, & \bPsi_t \sim \beta_m
  \Big(\frac{\lambda h_{t-1}}{2},
  \frac{(1-\lambda) h_{t-1}}{2} \Big), \\
  \bT_{t-1} = \chol{} \; \bX_{t-1}, \\
  h_{t-1} = \lambda h_{t-2} + 1.
\end{cases}
\]
As noted by Prado and West, $\lambda$ is constrained ``to maintain a valid
model, since we require either $h_t > m - 1$ or $h_t$ be integral [and less than
$m$].  The former constraint implies that $\lambda$ cannot be too small,
$\lambda > (m-2) / (m-1)$ defined by the limiting value.''  If $h_t = h$ is
integral, then $\lambda = (h-1) / h$.  Thus, one cannot pick from a range of
$\lambda$ for $h \leq m-1$, when $h$ must be an integer, which means there are
only $m-2$ allowed values of $\lambda$ for $\lambda \leq (m-2) / (m-1)$, the
limiting value when $h_t$ is sufficiently large.  When $m$ is large this is a
severe restriction.  Further, for $m \geq 2$, $h=2$ produces the smallest
possible value of $\lambda$; thus, $\lambda$ cannot be below $1/2$ unless $m=1$.
(We have replaced Prado and West's $\beta$ by $\lambda$ and their $q$ by $m$.)
The parameter $\lambda$ is important since it controls how much the model
smooths observations when forming estimates and one-step ahead predictions; thus
the constraints on $\lambda$ are highly undesirable.  In contrast, our model
lets $\lambda$ take on any value.

Given Model \ref{eqn:the-model}, we can derive several useful propositions.  The
proofs of these propositions, which synthesize and add to results from
\cite{uhlig-1994}, \cite{muirhead-1982}
\cite{diaz-garcia-gutierrez-jaimez-1997}, are technical, and hence we defer
their presentation to Section \ref{sec:proofs}.  Presently, we focus on the
\emph{closed form} formulas that one may make use of when forward filtering,
backward sampling, predicting one step into the future, and estimating $n$, $k$,
and $\lambda$.

First, some notation: inductively define the collection of data $\mcD_t =
\{\bY_t\} \cup \mcD_{t-1}$ for $t=1, \ldots, T$ with $\mcD_0 = \{\bSig_0\}$
where $\bSig_0$ is some covariance matrix.  Let the prior for $(\bX_1| \mcD_0)$
be $W_m(n, (k \bSigma_0)^{-1} / \lambda)$
where $W_m(d, \mathbf{V})$ is the Wishart distribution with $d \in \{1, \ldots,
m-1\} \cap (m-1,\infty)$ degrees of freedom and scale matrix $\mathbf{V}$.  (See
Definition \ref{def:wishart} for details.)  In the following, we implicitly
condition on the parameters $n$, $k$, and $\lambda$.

\begin{proposition}[Forward Filtering]
\label{fact:forward-filter}
Suppose $(\bX_t | \mcD_{t-1}) \sim W_m(n, (k \bSigma_{t-1})^{-1} / \lambda)$.
After observing $\bY_t$, the updated distribution is
\[
(\bX_t | \mcD_{t}) \sim W_m(k + n, (k \bSigma_t)^{-1})
\]
where 
\[
\bSigma_t = \lambda \bSigma_{t-1} + \bY_t.
\]
Evolving $\bX_t$ one step leads to
\[
(\bX_{t+1} | \mcD_t) \sim W_m(n, (k \bSigma_t)^{-1} / \lambda).
\]
\end{proposition}

\begin{proposition}[Backward Sampling]
\label{fact:backward-sample}
The joint density of $(\{\bX_t\}_{t=1}^T | \mcD_T)$ can be decomposed as
\[
p(\bX_T | \mcD_T) \prod_{t=1}^{T-1} p(\bX_{t} | \bX_{t+1}, \mcD_t)
\]
(with respect to the product measure on the $T$-fold product of $\sym_{m}^+$
embedded in $\bbR^{m(m+1)/2}$ with Lebesgue measure) where the distribution of
$(\bX_t | \bX_{t+1}, \mcD_t)$ is a shifted Wishart distribution
\[
(\bX_t | \bX_{t+1}, \mcD_{t})  = \lambda \bX_{t+1} + \bZ_{t+1}, \; \bZ_{t+1} \sim W_m(k, (k
\bSigma_{t})^{-1}).
\]
\end{proposition}

\begin{proposition}[Marginalization]
\label{fact:marginalization}
The joint density of of $\{\bY_t\}_{t=1}^T$ is given by
\[
p(\{\bY_t\}_{t=1}^T | \mcD_0) = \prod_{t=1}^T p(\bY_t | \mcD_{t-1}) 
\]
with respect to the differential form $\bigwedge_{t=1}^T (d\bY_t)$ where
$(d\bY_t)$ is as found in Definition \ref{def:wishart} for either the
rank-deficient or full-rank cases, depending on the rank of $\bY_t$.
(Differential forms, otherwise known as $K$-forms, are vector fields that may be
used to simplify multivariate analysis.  In particular, one may define densities
with respect to differential forms.  \cite{mikusinski-taylor-2002} provide a
good introduction to differential forms while \cite{muirhead-1982} shows how to
use them for statistics.) The density $p(\bY_t | \mcD_{t-1})$ is
\[
\pi^{-(mk - k^2)/2}
\frac{\Gamma_m(\frac{\nu}{2})}{\Gamma_m(\frac{n}{2})\Gamma_m(\frac{k}{2})}
\frac{|\bL_t|^{(k-m-1)/2} |\bV_t|^{n/2} }{|\bV_{t} + \bY_t|^{\nu/2}}
\]
with respect to $(d\bY_t)$ in the rank-deficient case and is
\[
\frac{\Gamma_m(\frac{\nu}{2})}{\Gamma_m(\frac{n}{2})\Gamma_m(\frac{k}{2})}
\frac{|\bY_t|^{(k-m-1)/2} |\bV_t|^{n/2} }{|\bV_{t} + \bY_t|^{\nu/2}}
\]
with respect to $(d\bY_t)$ in the full-rank case, where $\nu = n + k$, and \(
\bV_t = \lambda \bSigma_{t-1} \) with $\bSigma_{t} = \lambda \bSigma_{t-1} +
\bY_t$ like above.
\end{proposition}

Examining the one-step ahead forecasts of $\bY_t$ elucidates how the model
smooths.  Invoking the law of iterated expectations, one finds that
$\bbE[\bY_{t+1} | \mcD_{t}] = \bbE[\bX_{t+1}^{-1} | \mcD_{t}]$.  Since
$(\bX_{t+1}^{-1} | \mcD_{t})$ is an inverse Wishart distribution, its
expectation is proportional to $\bSigma_t$.  Solving the recursion for
$\bSigma_t$ from Fact \ref{fact:forward-filter} shows that
\begin{equation}
\label{eqn:Sigma}
\bSigma_t = \sum_{i=0}^{t-1} \lambda^i \bY_{t-i} + \lambda^t \bSigma_0.
\end{equation}
Thus, the forecast of $\bY_{t+1}$ will be a scaled, geometrically weighted average
of the previous observations.  If, further, one enforces the constraint
\begin{equation}
\label{eqn:constraint}
 \frac{1}{\lambda} = 1 + \frac{k}{n-m-1}
\end{equation}
then taking a step from $\bX_{t}$ to $\bX_{t+1}$ does not change its harmonic
average, that is $\bbE[\bX_{t}^{-1} | \mcD_{t}] = \bbE[\bX_{t+1}^{-1} | \mcD_{t}]$.  It
also implies that the one-step ahead point forecast of $(\bY_{t+1} | \mcD_{t})$ is
\begin{equation}
\label{eqn:constrained-forecast}
\bbE[\bY_{t+1} | \mcD_{t}] = (1-\lambda) \bSigma_{t} = (1 - \lambda) \sum_{i=0}^{t-1}
\lambda^i \bY_{t-i} + (1-\lambda) \lambda^{t} \bSigma_0.
\end{equation}
Hence in the constrained case, the one-step ahead forecast is the geometrically
weighted average of past observations.  For a geometrically weighted average,
the most recent observations are given more weight as $\lambda$ decreases.  It
has been known for some time that such averages provide decent one-step ahead
forecasts \citep{brown-book-1959}.

\section{Example: Covariance Forecasting}
\label{sec:example}

As noted initially, Model \ref{eqn:the-model} is an extension of the one
proposed by \cite{uhlig-1997}.  For the original model, when $k=1$, one might
consider observing a vector of heteroskedastic asset returns $\br_t \sim N(0,
\bX_t^{-1})$ where the precision matrix $\bX_t$ changes at each time step. The
extended model allows the precision matrix to change less often than the
frequency with which the returns are observed.  For instance, one may be
interested in estimating the variance of the daily returns, assuming that the
variance only changes from day to day, using \emph{multiple} observations taken
from \emph{within} the day.

To that end, suppose the vector of intraday stock prices evolves as geometric
Brownian motion so that on day $t$ the $m$-vector of log prices is
\[
\bp_{t,s} = \bp_{t,0} + \bmu s + \bV_{t}^{1/2} (\bw_{t+s} - \bw_t)
\]
at time $s$, where $s$ the fraction of the trading day that has elapsed,
$\{\bw_s\}_{s \geq 0}$ is an $m$-dimensional Brownian motion, and $\bV_t^{1/2}
{\bV_t^{1/2}}' = \bX^{-1}_{t}$.  In practice, $\bmu$ is essentially zero, so we
will ignore that term.  Further, suppose one has viewed the vector of prices at
$k+1$ equispaced times throughout the day so that \( \br_{t,i} = \bp_{t,i} -
\bp_{t,i-1/k}, \; i = 1/k, \ldots, 1.  \) Then \( \br_{t,i} \sim N(0, \bX_t^{-1}
/ k) \) and $\bY_t = \sum_{i=1}^k \br_{t,i} \br_{t,i}'$ is distributed as
$W_m(k, \bX_t^{-1}/k)$.  Letting $\bX_{t} = \mcU(\bX_{t-1})' \bPsi_{t}
\mcU(\bX_{t-1})$ where $\mcU(\cdot)$ computes the upper Cholesky decomposition
and $\bPsi_{t} \sim \beta_m(n/2, k/2)$, we recover Model \ref{eqn:the-model}
exactly.  Of course, in reality, returns are not normally distributed; they are
heavy tailed and there are diurnal patterns within the day.  Nonetheless, the
realized covariance literature, which we discuss in more detail in Section
\ref{sec:realized-covariance-matrices}, suggests that taking $\bY_t$ to be an
estimate of the daily variance $\bX_t^{-1}$ is a reasonable thing to do; though
to suppose that the error is Wishart is a strong assumption.  More dubious is
the choice of the evolution equation for $\bX_t$; a point we discuss further in
Section \ref{sec:discussion}.  But the evolution equation for
$\{\bX_t\}_{t=1}^T$ does provide a likelihood that accommodates closed form
forward filtering and backward sampling formulas, and possesses only a few
parameters, which makes it a relatively cheap model to employ.

The one mild challenge when applying the model is estimating $\bSigma_0$.
However, it is possible to ``cheat'' and not actually estimate $\bSig_0$ at all.
Consider (\ref{eqn:Sigma}) and ponder the following two observations.  First,
$\bSig_t$ is a geometrically weighted sum in $\{\bSig_0, \bY_1, \ldots,
\bY_t\}$.  Second, the least important term in the sum is $\bSig_0$.  Thus, one
can reasonably ignore $\bSig_0$ if $t$ is large enough.  To that end, we suggest
setting aside the first $\tau_1$ observations and using $\{\bSig_{\tau_1},
\bY_{\tau_1+1}, \ldots, \bY_{\tau_2}\}$ where $\bSig_{\tau_1} =
\sum_{i=0}^{\tau_1} \lambda^i \bY_{\tau_1-i}$ to learn $n$, $k$, and $\lambda$
using Proposition \ref{fact:marginalization} and the prior $p(X_{\tau_1+1} |
\mcD_{\tau_1}) \sim W_m(n, (k \Sigma_{\tau_1})^{-1} / \lambda)$.  It may seem
costly to disregard the first $\tau_1$ observations, but since there are so few
parameters to estimate this is unlikely to be a problem---the remaining data
will suffice.

\begin{figure}
\begin{center}
\includegraphics[scale=0.4]{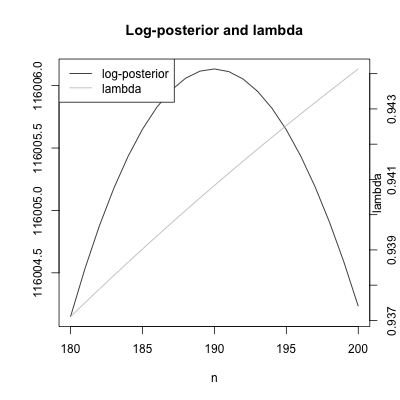}
\includegraphics[scale=0.5]{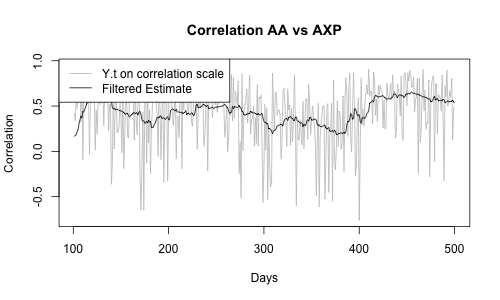}
\caption{\label{fig:wishart-dlm-llh-30} Level sets of the log posterior and
  filtered estimates on the correlation scale.}
\end{center}
On the left is the log posterior of $n$ calculated using $\{\bSigma_{50},
\bY_{51}, \ldots, \bY_{100}\}$ and constraint (\ref{eqn:constraint}).  The black
line is the log posterior in $n$, which has a mode at $n=190$ corresponding to
$\lambda = 0.94$.  The grey line is $\lambda$ as a function of $n$.  On the
right are the values of $\bY_t$ and the estimate $\bbE[\bX_t^{-1} | \mcD_t]$ for
$t=101, 500$ on the correlation scale.  A truncated time series was used to
provide a clear picture.
\end{figure}

This is the process used to generate Figure \ref{fig:wishart-dlm-llh-30} (with
$\tau_1 = 50$ and $\tau_2 = 100$).  The data set follows the $m=30$ stocks that
comprised the Dow Jones Industrial Average in October, 2010.  Eleven intraday
observations were taken every trading day for almost four years to produce 927
daily, rank-10 observations $\{\bY_{t}\}_{t=1}^{927}$.  Since the observations
are rank-deficient, we know that $k=10$.  (In the full-rank case, we will
estimate $k$.)  We constrain $\lambda$ using (\ref{eqn:constraint}) so that the
only unknown is $n$.  Given an improper flat prior for $n > 29$, the posterior
mode is $n=190$, implying that $\lambda = 0.94$, a not unusual value for
exponential smoothing.  Once $n$ is set, one can filter forward, backward
sample, and generate one-step ahead predictions in closed form.  The right side
of Figure \ref{fig:wishart-dlm-llh-30} shows the filtered covariance between
Alcoa Aluminum and American Express on the correlation scale.

However, one need not take such a literal interpretation of the model.  Instead
of trying to justify its use on first principles, one may simply treat it as a
covariance-valued state-space model, which we do presently.  As noted in the
introduction and elaborated on in Section
\ref{sec:realized-covariance-matrices}, realized covariance matrices are good
estimates of the daily covariance matrix of a vector of financial asset returns.
Since realized covariance matrices are good estimates it is natural to try to
use them for prediction.  The statistics themselves place very few restrictions
on the distribution of asset prices and their construction is non-parametric.
In other words, the construction of a realized covariance matrix (at least the
construction we use) says little about the evolution of the latent daily
covariance matrices.

\begin{figure}
\begin{center}
\includegraphics[scale=0.4]{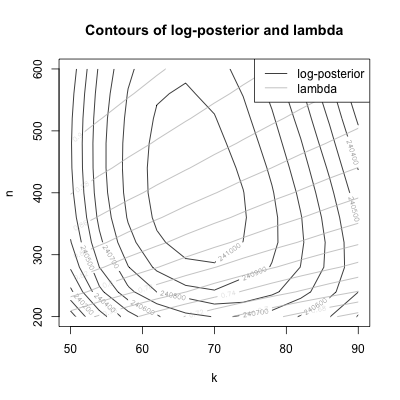}
\includegraphics[scale=0.5]{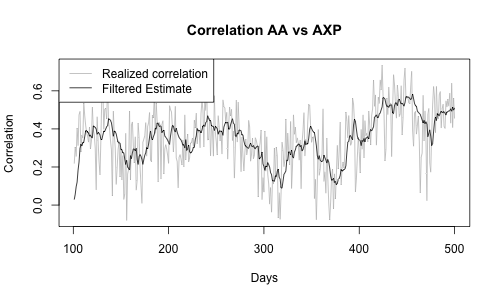}
\caption{\label{fig:wishart-dlm-llh} Level sets of the log posterior and
  filtered estimates on the correlation scale.}
\end{center}
On the left is the log posterior of $(k,n)$ calculated using $\{\bSigma_{50},
\bY_{51}, \ldots, \bY_{100}\}$ and constraint (\ref{eqn:constraint}).  The
black line is the level set of the log posterior as a function of $(k,n)$, which
has a mode at $(67, 396)$ corresponding to $\lambda = 0.85$.  The grey line is
the level set of $\lambda$ as a function of $(k,n)$.  On the right are the
values of $\bY_t$ and the estimate $\bbE[\bX_t^{-1} | \mcD_t]$ for $t=101, 500$
on the correlation scale.
\end{figure}

But we do not need to know the exact evolution of the latent daily covariance
matrices to employ Model \ref{eqn:the-model} to make short-term predictions.  To
that end, we may treat realized covariance matrices $\{\bY_t\}_{t=1}^T$ as
$\sym_{m}^+$-valued data that track the latent daily covariances
$\{\bX_{t}^{-1}\}_{t=1}^T$.  We construct these realized statistics using the
same $m=30$ stocks over the same time period as above, but using all of the
intraday data, which results in full-rank observations (see Section
\ref{sec:thedata} for details).  We follow the same basic procedure as above to
estimate $k$ and $n$, and implicitly $\lambda$ by constraint
(\ref{eqn:constraint}).  Selecting an improper flat prior for $n > 29$ and $k >
29$ yields the log-posterior found in Figure \ref{fig:wishart-dlm-llh}.  The
posterior mode is at $(67, 396)$ implying $\lambda = 0.85$.  The gray lines in
Figure \ref{fig:wishart-dlm-llh} correspond level sets of $\lambda$ in $k$ and
$n$.  As seen in the figure, the uncertainty in $(k,n)$ is primarily in the
direction of the steepest ascent of $\lambda$.  One can use Proposition
\ref{fact:marginalization} and the method of generating $\Sigma_{\tau_1}$ to
construct a random walk Metropolis sampler as well.  Doing that we find the
posterior mean to be $(67, 399)$, which implies an essentially identical
$\lambda$.  A histogram of the posterior of $\lambda$ is in Figure
\ref{fig:lambda}, showing that, though the the direction of greatest variation
in $(k,n)$ corresponds to changes in $\lambda$, the subsequent posterior
standard deviation of $\lambda$ is small.

Recall, our original motivation for studying $\sym_{m}^+$-valued state-space
models was the observation that exponentially smoothing realized covariance
matrices generates better one-step ahead predictions than factor stochastic
volatility.  In those initial experiments, we used cross-validation to pick the
smoothing parameter $\lambda$.  Figure \ref{fig:lambda} shows that one arrives
at the same conclusion under two different measures of performance using Model
\ref{eqn:the-model} and the method described previously to pick $\lambda$.  The
mesures are described in the caption to Figure \ref{fig:lambda}.  To summarize:
it is better to use our simple $\sym_{m}^+$-valued state-space model with
realized covariance matrices to make short term predictions than it is to use
factor stochastic volatility with only daily returns.

Though we pick a point estimate of the system's parameters above and then fix
that value to make predictions, one can operate in a fully Bayesian manner when
computing one-step ahead predictions as well as filtered distributions and the
posterior distribution of the latent states.  In particular, one can sample the
posterior distribution $p(\theta | \mcD_t)$, $\theta = (k,n)$, and then use
those posterior samples to draw from $p(\bX_t | \mcD_t, \theta)$ or $p(\bX_{t+1}
| \mcD_t, \theta)$ using forward filtering or to draw from $p(\{\bX_s\}_{s=1}^t
| \mcD_t, \theta)$ using forward filtering and backward sampling to get the
corresponding joint distributions $p(\bX_t, \theta | \mcD_t)$, $p(\bX_{t+1},
\theta | \mcD_t)$, and $p(\{\bX_s\}_{s=1}^t, \theta | \mcD_t)$.

\begin{figure}
\begin{center}
\raisebox{-.5\height}{%
\includegraphics[scale=0.4]{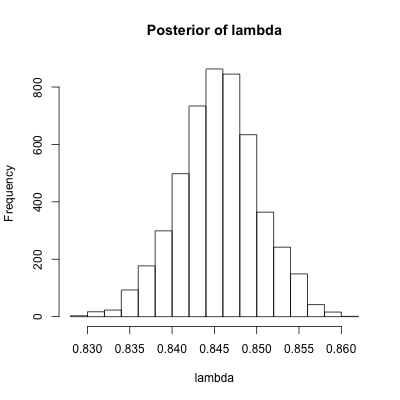}
}
\begin{tabular}{l | l l }
Model & MVP & PLLH \\
\hline
FSV 1 & 0.01021 & 94910 \\
FSV 2 & 0.00992 & 95155 \\
UE    & 0.00929 & 96781 \\
\multicolumn{3}{c}{} \\
\multicolumn{3}{l}{MVP: lower is better} \\
\multicolumn{3}{l}{PLLH: higher is better}
\end{tabular}
\caption{\label{fig:lambda} The posterior of $\lambda$ and the performance of
  Model \ref{eqn:the-model}.}
\end{center}
On the left: the posterior of $\lambda$ calculated using $\{\bSigma_{50},
\bY_{51}, \ldots, \bY_{100}\}$, constraint (\ref{eqn:constraint}), and posterior
samples of $(k,n)$.  On the right: the performance of Model \ref{eqn:the-model}
versus factor stochastic volatility using 1 and 2 factors.  ``MVP'' stands for
minimum variance portfolios and ``PLLH'' stands for predictive log-likelihood.
For all of the models, a sequence of one-step ahead predictions of the latent
covariance matrices $\{ \hat \bX_t^{-1} \}_{t=101}^{920}$ was generated.  For
Model \ref{eqn:the-model}, we set $\lambda$ to be the posterior mode found from
the data $\{\bSigma_{50}, \bY_{51}, \ldots, \bY_{100}\}$, as described in
Section \ref{sec:example}, to generate the one-step ahead predictions.  For
factor stochastic volatility, we picked the point estimate $\hat \bX_t^{-1}$ to
be an approximation of the mean of $(\bX_t^{-1} | \mcF_{t-1})$ where $\mcF_t =
\{\br_1, \ldots, \br_{t}\}$ and $\br_t$ is the vector of open to close
log-returns on day $t$.  For the MVP column, the one-step ahead predictions were
used to generate minimum variance portfolios for $t=101, \ldots, 920$.  The
column reports the empirical standard deviation of the subsequent portfolios.
Lower standard deviation is better.  For the PLLH column, the one-step ahead
predictions were used to calculate the predictive log-likelihood
$\sum_{i=101}^{920} \log \phi(\br_t; 0, \hat \bX_t^{-1})$ where $\phi$ is a
multivariate Gaussian kernel.  A higher predictive log-likelihood is better.
Model \ref{eqn:the-model} does better on both counts.
\end{figure}

\section{Discussion}
\label{sec:discussion}

Employing exponentially weighted moving averages to generate short-term
forecasts is not new.  These methods were popular at least as far back as the
first half of the 20th century \citep{brown-book-1959}.  In light of this, it
may seem that Model \ref{eqn:the-model} is rather unglamorous.  But this is only
because we have explicitly identified how the model uses past observations to
make predictions.  In fact, many models of time-varying variance behave
similarly.  For instance, GARCH \citep{bollerslev-1986} does exponential
smoothing with mean reversion to predict the variance using square returns.
Stochastic volatility \citep{taylor-1982} does exponential smoothing with mean
reversion to predict the log variance using log square returns.  Models that
include a leverage effect do exponential smoothing so that the amount of
smoothing depends on the direction of the returns.  Thus, it should not be
surprising or uninteresting when a state-space model generates predictions with
exponential smoothing or some variation thereof.

This helps explains why a simple model (\ref{eqn:the-model}) with high-quality
observations can generate better short-term predictions than a complicated model
(factor stochastic volatility) with low-quality data.  First, both models, in
one way or another, are doing something similar to exponential smoothing.
Second, the true covariance process seems to revert quite slowly.  Thus, there
will not be much difference between a one-step ahead forecast that lacks mean
reversion (Model UE) and a one-step ahead forecast that includes mean reversion
(factor stochastic volatility).  Since the prediction mechanisms are similar,
the model that uses a ``higher resolution'' snapshot of the latent covariance
matrices has the advantage.  Of course, these observations only apply when using
factor stochastic volatility with daily returns.  It may be the case that one
can use intraday information along with some specialized knowledge about the
structure of market fluctuations (like factor stochastic volatility) to generate
better estimates and predictions.

Despite Model UE's short-term forecasting success, it does have some faults.
First, the evolution of $\{\bX_t\}_{t=1}^T$ can be rather degenerate.  In the
one-dimensional case, when $\{\log x_t\}_{t=1}^\infty$ is not a martingale,
$\{x_t\}_{t=1}^\infty$ either almost surely converges to zero or almost surely
diverges.  (The discussion before Proposition \ref{prop:backward-sample} expands
upon this point.)  Presumably, the multivariate case suffers from something
similar and, clearly, this does not reflect the dynamics we want to capture.
Second, its $k$-step ahead predictions do not revert to some mean, which is what
one would expect when modeling a stationary process.  In fact, the first point
suggests that things are worse than that: the $k$-step ahead predictive
distributions may degenerate.  Consequently, the model will perform poorly as
the horizon of the prediction increases.

Thus, to the larger question, ``How does one construct rich, tractable
state-space models on curved spaces,'' we only gave a partial answer, showing
how to create a tractable model---one in which the densities of interest may be
computed and sampled.  In essence, descriptive richness was sacrificed for
tractability.  One might proceed in the opposite direction by endowing the
latent process with rich dynamics initially.  For instance, one may transform a
positive-definite matrix $\bX$ into an unconstrained planar coordinate system
using the factorization $\bU' \exp(\bD)\bU = \bX$, where $\bU$ is upper
triangular, $\bD$ is diagonal, and $\exp$ is the matrix exponential, and then
model the dynamics in the planar coordinates $(\bU, \bD)$.  But, in that case,
one must deal with a potentially inconvenient distribution in the $\bX$
coordinates for forward filtering or backward sampling.  Comparing the benefits
of each approach within the context of Bayesian state-space inference is left to
future work.


\section{Technical Details}
\label{sec:proofs}


Much of the calculus one needs can be found \cite{uhlig-1994} or
\cite{muirhead-1982}.  We synthesize those results here.  We are not aware of
results in either regarding backward sampling or marginalization.

First, some notation: Assume $k,m \in \bbN$, $k \leq m$.  Let $\sym^+_{m,k}$
denote the set of positive semi-definite symmetric matrices of rank $k$ and
order $m$.  When $k = m$, we drop $k$ from the notation so that $\sym^+_m$
denotes the set of positive-definite symmetric matrices of order $m$.
For symmetric matrices $\bA$ and $\bB$, let $\bA < \bB$ denote $\bB - \bA \in
\sym_m^+$ 
.  For $\bA \in \sym^+_m$ let
\[
\{\sym^+_{m,k} < \bA\} = \{\bC \in \sym^+_{m,k}: \bC < \bA\}.
\]
If $k > m-1$ is real and we write $\sym_{m,k}^+$ then we implicitly mean
$\sym_{m}^+$.  We will use $|\cdot|$ to denote the determinant of a matrix and
$\bI$ to denote the identity.  We at times follow \cite{muirhead-1982} and
define densities with respect to differential forms (also known as $K$-forms or
differential $K$-forms).  \cite{mikusinski-taylor-2002} is a good introduction
to calculus on manifolds.  The handouts of \cite{edelman-handouts-2005} provide
a more succinct introduction.

\begin{definition}[Wishart distribution] 
  \label{def:wishart}
  A positive semi-definite symmetric matrix-valued random variable $\bY$ has
  Wishart distribution $W_m(k, \bV)$ for $k \in \bbN$ and $\bV \in \sym_m^+$ if
  \[
  \bY \sim \sum_{i=1}^k \br_i \br_i', \; \; \br_i \stackrel{iid}{\sim} N(0, \bV), \; \; i
  = 1, \ldots, k.
  \]
  When $k > m - 1$, the density for the Wishart distribution is
  \[
  \frac{|\bY|^{(k-m-1)/2}}{2^{mk/2}
    \Gamma_m \big(\frac{k}{2} \big) |\bV|^{k/2}}
  \exp \Big(\tr - \frac{1}{2} \bV^{-1} \bY \Big)
  \]
  \citep{muirhead-1982} with respect to the volume element
  \[
  (d\bY) = \bigwedge_{1 \leq i \leq j \leq m} d\bY_{ij}.
  \]
  When $k \leq m - 1$ and $\bY$ is rank deficient, the density is
  \[
  \frac{\pi^{-(mk - k^2)/2} |\bL|^{(k-m-1)/2}}{2^{mk/2}
    \Gamma_k\big(\frac{k}{2}\big) |\bV|^{k/2}}
  \exp \Big(\tr - \frac{1}{2} \bV^{-1} \bY \Big)
  \]
  with respect to the volume element
  \[
  (d\bY) = 2^{-k} \prod_{i=1}^k l_i^{m-k} \prod_{i < j}^{k} (l_i - l_j) (\bH_1'd\bH_1)
  \wedge \bigwedge_{i=1}^k dl_i
  \]
  where $Y = \bH_1 \bL \bH_1'$, $H_1$ is a matrix of orthonormal columns of
  order $m \times k$, and $\bL = \text{diag}(l_1, \ldots, l_k)$ with decreasing
  positive entries \citep[Thm.\ 6]{uhlig-1994}.  The notation $(\bH_1' d\bH_1)$
  is shorthand for a differential $K$-form from the Steifel manifold $V_{m,k}$
  embedded in $\bbR^{m \times k}$ where $K = mk - k(k+1)/2$ \citep[p.\
  63]{muirhead-1982}.  One can extend the definition of the Wishart distribution
  to real values of $k > m-1$ for $\sym_m^+$-valued random variables by defining
  $Y \sim W_m(k, \bV)$ to have the full rank density defined above.
\end{definition}

\begin{definition}[the bijection $\tau$]
\label{def:tau}
Assume $m \in \bbN$ and $k \in \{1, \ldots, m\}$.  A single bijection provides
the key to both the evolution of $\bX_t$ in Model \ref{eqn:the-model} and to the
definition of the beta distribution.  In particular, let $\tau : \sym_{m,k}^+
\times \sym_{m}^+ \ra \sym_{m}^+ \times \{\sym_{m,k}^+ < \bI\}$ take $(\bA,\bB)$
to $(\bS,\bU)$ by letting $\bT'\bT = \bA + \bB$ be the Cholesky factorization of
$\bA + \bB$ and letting
\[
\begin{cases}
\bS = \bA + \bB, \\
\bU = {\bT^{-1}}' \bA \bT^{-1}.
\end{cases}
\]
Conversely, let $g : \sym_{m}^+ \times \{\sym_{m,k}^+ < \bI \} \ra \sym_{m,k}^+ \times
\sym_{m}^+$ take $(\bS,\bU)$ to $(\bA,\bB)$ by letting $\bT'\bT = \bS$ be the
Cholesky decomposition of $\bS$ and
\[
\begin{cases}
\bA = \bT' \bU \bT, \\
\bB = \bT' (\bI - \bU) \bT.
\end{cases}
\]
One can see that $g$ is the inverse of $\tau$ since $\tau(g(\bS,\bU)) = (\bS,\bU)$ and
$g(\tau(\bA,\bB)) = (\bA,\bB)$.
\end{definition}


\begin{definition}[beta distribution]
  \label{def:beta-distribution}
  Let $\bA \sim W_m(k, \bSig^{-1})$ and $\bB \sim W_m(n, \bSig^{-1})$ be
  independent where $n > m-1$ and either $k < m$ is an integer or $k > m-1$ is
  real-valued.  Let $(\bS,\bU) = \tau(\bA,\bB)$.  The beta distribution,
  $\beta_m(n/2, k/2)$, is the distribution of $\bU$.  When $k < m$ is an
  integer, the beta distribution $\beta_m(k/2, n/2)$ is the distribution of $\bI
  -\bU$ where $\bU \sim \beta_m(n/2, k/2)$. (See Definition 1 from
  \cite{uhlig-1994} and p.\ 109 in \cite{muirhead-1982}.)
\end{definition}

The following theorem synthesizes results from \cite{uhlig-1994},
\cite{muirhead-1982}, and \cite{diaz-garcia-gutierrez-jaimez-1997}.

\begin{theorem}
  Based on \cite[Thm.\ 3.3.1]{muirhead-1982}, \cite[Thm.\ 7]{uhlig-1994} and
  \cite[Thm.\ 2]{diaz-garcia-gutierrez-jaimez-1997}
  \label{thm:muirhead}.  Let $n > m-1$ and let either $k < m$ be an integer or
  $k > m-1$ be real-valued.  The bijection $\tau : \sym_{m,k}^+ \times
  \sym_{m}^+ \ra \sym_{m}^+ \times \{\sym_{m,k}^+ < \bI\}$ from Definition
  \ref{def:tau} changes
\begin{equation}
\label{muirhead-ab-dist}
\bA \sim W_m(k, \bSigma^{-1}) \perp \bB \sim W_m(n, \bSigma^{-1})
\end{equation}
to
\begin{equation}
\label{muirhead-su-dist}
\bS \sim W_m(n+k, \bSigma^{-1}) \perp \bU \sim \beta_m(k/2, n/2).
\end{equation}
\end{theorem}

\begin{proof}
  Thm.\ 3.3.1 in \cite{muirhead-1982} proves this in the full rank case. Thm.\ 7
  in \cite{uhlig-1994} proves this in the rank 1 case.  Thm.\ 2 in
  \cite{diaz-garcia-gutierrez-jaimez-1997} proves it in the general rank
  deficient case.
\end{proof}

Theorem \ref{thm:muirhead} justifies forward filtering in Model
\ref{eqn:the-model} as follows.

\begin{proof}[Forward Filtering]
  Suppose we start at time $t-1$ with data $\mcD_{t-1}$, so that the joint
  distributions of $\bX_{t-1}$ and $\bPsi_t$ is characterized by
\[
\bX_{t-1} \sim W_m(n+k, (k \bSigma_{t-1})^{-1}) \perp (\bI-\bPsi_t) \sim
\beta_m(k/2,n/2),
\]
which looks like (\ref{muirhead-su-dist}).  Theorem \ref{thm:muirhead} shows
that the bijection $\tau^{-1}$ takes $(\bX_{t-1}, \bI-\bPsi_t)$ to
\[
\bZ_t \sim W_m(k, (k \bSigma_{t-1})^{-1}) \perp \lambda \bX_t \sim W_m(n, (k
\bSigma_{t-1})^{-1}),
\]
which is (\ref{muirhead-ab-dist}) after applying the transformation summarized
by
\begin{equation}
\label{eqn:example-transformation}
\begin{cases}
\bZ_t = \bT_{t-1}' (\bI - \bPsi_t) \bT_{t-1}, \\
\lambda \bX_t = \bT_{t-1}' \bPsi_t \bT_{t-1}, \\
\bX_{t-1} = \bZ_t + \lambda \bX_t, \\
\bT_{t-1} = \chol{} \; \bX_{t-1}.
\end{cases}
\end{equation}
The transformation includes the evolution equation in (\ref{eqn:the-model})
since \( \bX_t = \bT_{t-1}' \bPsi_t \bT_{t-1} / \lambda \).  It also yields
$(\bX_t | \mcD_{t-1}) \sim W_m(n, (k \bSigma_{t-1})^{-1} / \lambda)$.  Conjugate
updating then yields $(\bX_t | \mcD_{t}) \sim W_m(n+k, (k \bSig_{t})^{-1})$
where $\bSig_{t} = \lambda \bSig_{t-1} + \bY_t$.
\end{proof}

The reader may notice that the choice of distribution for $\bPsi_{t}$ is
precisely the one that facilitates forward filtering.  In particular, assuming
that $(\bX_{t-1} | \mcD_{t-1})$ has an acceptable distribution to start, then
$(\bX_t | \mcD_{t-1})$ will have an acceptable distribution to update, so that
$(\bX_t | \mcD_t)$ will have a distribution that lets us play the game all over
again.  However, we cannot easily write down the distribution of $(\bX_{t+k} |
\mcD_{t})$ for anything but $k=0$ or $1$.  To see why, assume that we start at
time $t-1$ with data $\mcD_{t-1}$ and evolve to $(\bX_{t} | \mcD_{t-1}) \sim
W_m(n, (k \bSig_{t-1})^{-1} / \lambda)$ just like above.  Now consider moving
from $\bX_t$ to $\bX_{t+1}$ without updating:
\[
\begin{cases}
\bT_{t} = \chol \; \bX_{t} \\
\bX_{t+1} = \bT_{t}' \bPsi_t \bT_{t}, & \bPsi_t \sim \beta_m(n/2,k/2).
\end{cases}
\]
The distribution of $\bI - \bPsi_t$ is $\beta_m(n/2,k/2)$ but the distribution
of $(\bX_t|\mcD_{t-1})$ is $W_m(n, \ldots)$.  We cannot apply Theorem
\ref{thm:muirhead} at this point because there is a mismatch in the degrees of
freedom of $(\bX_{t} | \mcD_{t-1})$ and the parameters of $\bI - \bPsi_t$---we
need $n + k$ not $n$ degrees of freedom!  Thus, the distribution of $(\bX_{t+1}
| \mcD_{t-1})$ is unknown.

Despite not knowing its distribution, one can show that the evolution of
$\{\bX_{t}\}_{t=1}^T$ is rather degenerate.  To see this, consider the one
dimensional case, in which
\[
x_t = x_{t-1} \psi_t / \lambda, \; \psi_t \sim \beta(n/2,k/2).
\]
Following \cite{shephard-1994}, transforming this equation by the logarithm
yields
\[
w_t = w_{t-1} + \nu_t , \; \nu_t \sim \log( \beta(n/2,k/2) / \lambda )
\]
where $w_t = \log x_t$ and $\nu_t = \log (\psi_t/\lambda)$.  Let $m =
\bbE[\nu_t]$.  When $m \neq 0$, the law of large numbers says that for almost
every path there is some $t^*$ such that $t (m - |m|/2) < w_t < t (m + |m|/2)$
for $t > t^*$.  That is, the paths diverge.  Hence, when $m \neq 0$, the paths
of $x_t$ either converge to $0$ or diverge.  The same phenomenon can be seen
when numerically simulating data in the multivariate case.  This makes
generating synthetic data difficult since $\bX_t$ can quickly become numerically
singular.  It also implies that the predictive distributions are unruly.

\begin{proposition}
  \label{prop:backward-sample}
  Assume $\bS$ and $\bU$ are as in Theorem \ref{thm:muirhead} and let $(\bA,
  \bB) = \tau^{-1}(\bS, \bU)$.  Then the conditional distribution of $(\bS|\bB)$
  is
  \begin{equation}
    \label{eqn:muirhead-backward-sample}
    (\bS|\bB) = \bB + \bZ, \; \bZ \sim W_m(k, \bSigma^{-1}).
  \end{equation}
\end{proposition}

\begin{proof}
  Let $\bS$ and $\bU$ be as in Theorem $\ref{thm:muirhead}$ and let $(\bA, \bB)
  = \tau^{-1}(\bS, \bU)$.  Let $p$ be the rank of $\bA$.  Fix $\bB$ and define a
  change of variables $g$ by $\bA = \bS - \bB$.  Jointly, $(\bA,\bB)$ has a
  density with respect to the differential form $(d\bA) \wedge (d \bB)$ where
  $\bA$ is a $K$-form where $K = np - p(p-1)/2$:
  \[
  (d\bA) = \sum_{i_1 < \ldots < i_K} f_{i_1 < \ldots < i_K} (\bA) \; d\bA_{i_1} \wedge
  \cdots \wedge d\bA_{i_K}
  \]
  where the index of $d\bA$ corresponds to the vectorized (by column) upper
  triangular portion of $\bA$.  Under $g$, the pull back of $d\bA_i$ is
  \[
  g^*(d\bA_i) = d\bS_i;
  \]
  thus, 
  \[
  (d\bS) = \sum_{i_1 < \ldots < i_K} f_{i_1 < \ldots < i_K} (\bS -  \bB) \; 
  d\bS_{i_1} \wedge \cdots \wedge d\bS_{i_K},
  \]
  where, again, the index corresponds to the vectorized upper triangular
  portion.
  Let $f_A(\bA) f_B(\bB)$ be the density of $(\bA,\bB)$ with respect to the
  differential form $(d\bA) \wedge (d\bB)$.  Under $g$, the differential form
  corresponding to the density of $(\bA,\bB)$,
  \[
  f_A(\bA) f_B(\bB) \; (d\bA) \wedge (d\bB),
  \]
  becomes
  \[
  f_A(\bS - \bB) f_B(\bB) \; (d\bS) \wedge (d \bB)
  \]
  on the manifold
  \[
  \{(\bS, \bB) : \bS \in \sym_{m}^+, \; \bB \in \sym_m^+, \; \bS - \bB \in \sym_{m,k}^+ \}.
  \]
  We know that $f_B(\bB) (d\bB)$ is the differential form corresponding to the
  distribution of $\bB$, hence $f_A(\bS-\bB) (d\bS)$ describes the conditional
  distribution of $(\bS|\bB)$.  Doing another change of variables shows that
  $(\bS|\bB)$ is a shifted Wishart distribution, that is
  \[
  (\bS|\bB) = \bB + \bZ, \; \bZ \sim W_m(k, \bSig^{-1}).
  \]
\end{proof}

\begin{proof}[Backward Sampling]

  The Markovian structure of the model ensures that we can decompose the joint
  density of the latent states given $\mcD_T$ (and $n$, $k$, $\lambda$) as
  \[
  p(\bX_T | \mcD_T) \prod_{i=1}^{T-1} p(\bX_t | \bX_{t+1}, \mcD_t).
  \]
  (The density is taken with respect to product measure on the $T$-fold product
  of $\sym_{m}^+$ embedded in $\bbR^{m(m+1)/2}$ with Lebesgue measure).
  Applying Proposition \ref{prop:backward-sample} with $(\bX_{t-1}|\mcD_{t-1})$
  as $\bS$, $\bI - \bPsi_t$ as $\bU$, and $(\lambda \bX_t | \mcD_{t-1})$ as
  $\bB$, we find that the distribution of $(\bX_{t-1} \mid \bX_{t}, \mcD_{t-1})$
  is
  \[
  (\bX_{t-1}|\bX_{t},\mcD_{t-1}) = \lambda \bX_{t}+ \bZ_{t}, \; \bZ_{t} \sim W_m(k,
  (k \bSigma_{t-1})^{-1}).
  \]


\end{proof}

\begin{proof}[Marginalization]
  First, by conditioning we can express the density $p(\{\bY_t\}_{t=1}^T |
  \mcD_0)$ as 
  \[
  \prod_{t=1}^T p(\bY_t | \mcD_{t-1})
  \]
  with respect to the differential form $\bigwedge_{t=1}^T (d\bY_t)$ where
  $(d\bY_t)$ is as in Definition \ref{def:wishart}.

  Thus, we just need to derive the distribution of $(\bY_t | \mcD_{t-1})$.
  Assume that $n > m-1$ and that either $k < m$ is an integer or $k > m-1$ is
  real-valued.  Suppose that $(\bY_t | \bX_t) \sim W_m(k, (k\bX_t)^{-1})$ and
  $(\bX_t | \mcD_{t-1}) \sim W_m(n, (k \bV_t)^{-1})$ where $\bV_t = \lambda
  \bSig_{t-1}$.  Then the density for $(\bY_t | \mcD_{t-1})$ is
  \[
  \pi^{-(mk - k^2)/2}
  \frac{\Gamma_m(\frac{\nu}{2})}{\Gamma_m(\frac{n}{2})\Gamma_k(\frac{k}{2})}
  \frac{|\bL_t|^{(k-m-1)/2} |\bV_t|^{n/2} }{|\bV_{t} + \bY_t|^{\nu/2}}
  \]
  in the rank-deficient case and is 
  \[
  \frac{\Gamma_m(\frac{\nu}{2})}{\Gamma_m(\frac{n}{2})\Gamma_m(\frac{k}{2})}
  \frac{|\bY_t|^{(k-m-1)/2} |\bV_t|^{n/2} }{|\bV_{t} + \bY_t|^{\nu/2}}
  \]
  in the full-rank case, with respect to the differential form $(d\bY_t)$ is as
  found in Definition \ref{def:wishart} for either the rank-deficient or
  full-rank cases respectively.

  We will only prove the rank-deficient case, since the full-rank case is
  essentially identical.  Consider the joint density $p(\bY_t | \bX_t) p(\bX_t |
  \mcD_{t-1})$:
\begin{gather*}
 \frac{\pi^{-(mk - k^2)/2} |k \bX_t|^{k/2}}{2^{mk/2}
    \Gamma_k\big(\frac{k}{2}\big)} |\bL_t|^{(k-m-1)/2}
  \exp \Big(\frac{-1}{2} \tr k \bX_{t} \bY_t \Big) \\
\frac{|k \bV_{t}|^{n/2}}
{2^{nm/2} \Gamma_m(\frac{n}{2})}
|\bX_t|^{(n-m-1)/2} \exp \Big( \frac{-1}{2} \tr k \bV_{t} \bX_t \Big)
\end{gather*}
(where $\bY_t = \bH_t \bL_t \bH_t$, $\bL_t$ is a $k \times k$ diagonal matrix
with decreasing entries, and $\bH_t$ is in the Steifel manifold $V_{m,k}$)
with respect to $(d\bY_t) \wedge (d\bX_t)$, which is
\[
\pi^{-(mk - k^2)/2} \frac{|\bL_t|^{(k-m-1)/2}}{2^{km/2} \Gamma_k(\frac{k}{2})}
\frac{|k \bV_{t}|^{n/2}}{2^{nm/2} \Gamma_m(\frac{n}{2})}
\; k^{km/2} |\bX_t|^{(\nu-m-1)/2}
\exp \Big( \frac{-1}{2} \tr k \big(\bV_{t} + \bY_t \big) \bX_t \Big),
\]
$\nu = n + k$.  The latter terms are the kernel for a Wishart distribution in
$\bX_t$.  Integrating the kernel with respect to $\bX_t$ yields
\[
\frac{2^{\nu m/ 2} \Gamma_m(\frac{\nu}{2}) }{|k (\bV_{t} + \bY_t)|^{\nu/2}}.
\]
Hence the density of $(\bY_t | \mcD_{t-1})$ is
\[
\pi^{-(mk - k^2)/2}
\frac{\Gamma_m(\frac{\nu}{2})  k^{\nu m/2}}{\Gamma_m(\frac{n}{2})\Gamma_k(\frac{k}{2})}
\frac{|\bL_t|^{(k-m-1)/2} |\bV_t|^{n/2} }{|k(\bV_{t} + \bY_t)|^{\nu/2}}
\]
with respect to $(d \bY_t)$.
Factoring the $k$ in the denominator gives us
\[
\pi^{-(mk - k^2)/2}
\frac{\Gamma_m(\frac{\nu}{2})}{\Gamma_m(\frac{n}{2})\Gamma_k(\frac{k}{2})}
\frac{|\bL_t|^{(k-m-1)/2} |\bV_t|^{n/2} }{|\bV_{t} + \bY_t|^{\nu/2}}.
\]
\end{proof}

\appendix

\section{Realized Covariance Matrices}
\label{sec:realized-covariance-matrices}

Realized covariance matrices are symmetric positive-definite estimates of the
daily quadratic variation of a multidimensional continuous-time stochastic
process.  Within the context of financial time series, there is both theoretical
and empirical evidence to suggest that a realized covariance matrix can be
interpreted as an estimate of the covariance matrix of the open to close log
returns.

Imagine that the market in which the assets are traded is open 24 hours a day
and that we are interested in estimating the covariance matrix of daily log
returns.  Following \cite{bn2002}, let $\bp_s$ be the $m$-vector of log prices
where $s$ is measured in days and suppose that it is a Gaussian process of the
form
\[
\bp_s = \int_{0}^s {\bV_u^{1/2}} d \bw_u
\]
where $\{\bw_s\}_{s \geq 0}$ is an $m$-dimensional Brownian motion and
$\{{\bV_s^{1/2}}\}_{s}$ is a continuous, deterministic, symmetric positive
definite $m \times m$ process such that the square of $\{{\bV_s^{1/2}}\}_{s}$ is
integrable.  Then the day $t$ vector of log returns $\br_t = (\bp_t -
\bp_{t-1})$ is distributed as
\[
\br_t \sim N \Big( 0, \int_{t-1}^t \bV_u du \Big).
\]
where $\bV_u = {{\bV_u^{1/2}}} {{\bV_u^{1/2}}}'$.  The quadratic covariation matrix
(quadratic variation henceforth) measures the cumulative local (co)-fluctuations
of the sample paths:
\[
\ip{\bp}_s = \plim{|\Delta_N| \ra 0} \sum_{i=1}^{K_N} (\bp_{u_i} -
\bp_{u_{i-1}}) (\bp_{u_i} - \bp_{u_{i-1}})'.
\]
where the limit holds for any sequence of partitions of the form $\Delta_N =
\{u_0 = 0 < \ldots < u_{K_N} = s\}$ and $|\Delta_N| = \max \{ u_i - u_{i-1} : i
\in 1, \ldots, K_N \}$.  It is always the case, even when $\{\bV_s^{1/2}\}_s$ is
a stochastic process correlated with $\{\bw_s\}_s$, that
\[
\int_{t-1}^t \bV_{u} du =  \ip{\bp}_t - \ip{\bp}_{t-1}.
\]
(See Proposition 2.10 in \cite{karatzas-shreve-book-1991}.)  Thus, in the
Gaussian process case, the variance of $\br_t$ is related to the quadratic
variation by
\[
\var(\br_t) = \ip{\bp}_t - \ip{\bp}_{t-1}
\]
If the assets under consideration are traded frequently, then the day-$t$
partition of trading times $\Delta_t^* = \{u_0 = t-1 < \ldots < u_{K_t} = t\}$
has $|\Delta_t^*|$ near zero so that
\[
\widehat{\mathbf{RC}}_t = \sum_{i=1}^{K_t} (\bp_{u_i} - \bp_{u_{i-1}}) (\bp_{u_i} -
\bp_{u_{i-1}})',
\]
where the summation is over $\Delta_t^*$, is a good estimate of $\ip{\bp}_t -
\ip{\bp}_{t-1}$.  This is the realized covariance.

The same logic proceeds when $\{\bV_s^{1/2}\}_s$ is stochastic process that is
independent of the Brownian motion.  In that case, the only major change is
\[
\Big( \br_t \Big| \int_{t-1}^t \bV_u du \Big) = N \Big(0, \int_{t-1}^t \bV_u du \Big),
\]
that is the log returns are a mixture of normals, so that
\[
\var \Big( \br_t \Big| \int_{t-1}^t \bV_u du \Big) = \ip{\bp}_t - \ip{\bp}_{t-1}.
\]
Since $\widehat{\mathbf{RC}}_t$ is a good estimate of $\ip{\bp}_t -
\ip{\bp}_{t-1}$ regardless of $\{\bV_s\}_s$, so long as the assets are traded
often enough, one still has a good estimate of the daily conditional variance
despite the fact that $\{\bV_s^{1/2}\}_s$ is random.  The nice thing about
quadratic variation is that it is well-defined for any process that is a
semimartingale \citep[Thm.\ 4.47]{jacod-shiryaev-2003}.  In that sense, it is a
completely non-parametric statistic; though the derivations above do not
necessarily hold once $\{\bV_s\}_s$ is correlated with the underlying Brownian
motion.  Empirical work has shown that $\widehat{\mathbf{RC}}_{t=1}^T$ can be
used to estimate and forecast the variance of the daily returns in the
univariate case \citep{andersen-etal-2001, koopman-etal-2005} and the covariance
matrix of the vector of daily returns in the multivariate case \cite{liu-2009}.



We treat the realized covariances $\{\widehat{\mathbf{RC}}_t\}_{t=1}^T$ (or
rather a different, related approximation to $\ip{\bp}_t - \ip{\bp}_{t-1}$
called realized kernels) as the noisy observations $\{\bY_t\}_{t=1}^T$ in
Section \ref{sec:example} and then infer $n$, $k$, and $\lambda$ to generate
filtered estimates and one-step ahead predictions of the latent covariance
matrices $\{\bX_t^{-1}\}_{t=1}^T$.  \cite{bn-etal-2011} describe how to
construct the matrix valued data and we follow their general approach to produce
symmetric positive-definite valued data $\{\bY_t\}_{t=1}^{927}$ for 927 trading
days and 30 assets.  Details of the construction and the data can be found in
Section \ref{sec:thedata}.

\section{Construction of Realized Kernel and Data}
\label{sec:thedata}

The data set follows the thirty stocks found in Table \ref{table:stocks}, which
comprised the Dow Jones Industrial Average as of October, 2010.  The raw data
consists of intraday tick-by-tick trading prices from 9:30 AM to 4:00 PM
provided by the Trades and Quotes (TAQ) database through Wharton Research Data
Services\footnote{Wharton Research Data Services (WRDS) was used in preparing
  this paper.  This service and the data available thereon constitute valuable
  intellectual property and trade secrets of WRDS and/or its third-party
  suppliers.  } .  The data set runs from February 27, 2007 to October 29, 2010
providing a total of 927 trading days.

\begin{table}

\begin{center}
\tiny
\begin{tabular}{l | l | l | l | l}
  Alcoa (AA) &
  American Express (AXP) &
  Boeing (BA) &
  Bank of America   (BAC) &
  Caterpillar  (CAT) \\
  Cisco  (CSCO)* &
  Chevron (CVX) &
  Du Pont (DD) &
  Disney (DIS) &
  General Electric (GE) \\
  Home Depot (HD) &
  Hewlett-Packard (HPQ) &
  IBM (IBM) &
  Intel (INTC)* &
  Johnson \& Johnson (JNJ) \\
  JP Morgan (JPM) &
  Kraft (KFT) &
  Coca-Cola (KO) &
  McDonald's (MCD) &
  3M (MMM) \\
  Merk (MRK) &
  Microsoft (MSFT)* &
  Phizer (PFE) &
  Proctor \& Gamble (PG) &
  AT\&T   (T) \\
  Traveler's (TRV) &
  United Technologies (UTX) &
  Verizon (VZ) &
  Walmart (WMT) &
  Exxon Mobil (XOM)
\end{tabular}
\caption{\label{table:stocks} The thirty stocks that make up the data set.}
\end{center}
The asterisk denotes companies whose primary exchange is the NASDAQ.  All other companies
trade primarily on the NYSE.
\end{table}

Our construction of the realized kernels is based upon \cite{bn2009rkip,
  bn-etal-2011}.  Warning: we re-use the letters $X$ and $Y$, but now they refer
to vector-valued continuous-time processes!  Barndorff-Nielsen et al.'s model,
which takes into account market microstructure noise, is
\[
X_{t_i} = Y_{t_i} + U_{t_i}
\]
where $\{t_i\}_{i=1}^n$ are the times at which the $m$-dimensional vector of log
stock prices, $\{X_t\}_{t \geq 0}$, are observed, $\{Y_t\}_{t \geq 0}$ is the
latent log stock price, and $\{U_{t_i}\}_{t=1}^n$ are errors introduced by
market microstructure.  The challenge is to construct estimates of the quadratic
variation of $\{Y_t\}$ with the noisy data $\{X_{t_i}\}_{i=1}^n$.  They do this
using a kernel approach,
\[
K(X_t) = \sum_{h=-H}^H k \Big( \frac{h}{H} \Big) \Gamma_h
\]
where
\[
\Gamma_h (X_t) = \sum_{j=h+1}^n x_j x_{j-h}', \textmd{ for } h \geq 0,
\]
with $x_j = X_{s_j} - X_{s_{j-1}}$ and $\Gamma_h = \Gamma_{-h}'$ for $h < 0$.
The kernel $k(x)$ is a weight function and lives within a certain class of
functions.  While this provides a convenient formula for calculating realized
kernels, the choice of weight function and proper bandwidth $H$ requires some
nuance.  \cite{bn-etal-2011} discuss both issues.  We follow their suggestions,
using the Parzen kernel for the weight function and picking $H$ as the average
of the collection of bandwidths $\{H_i\}_{i=1}^m$ one calculates for each asset
individually.  Before addressing either of those issues one must address the
practical problem of cleansing and synchronizing the data.

\begin{description}

\item[Clean the data]:
The data was cleaned using the following rules.
\begin{itemize}
\item Retrieve prices from only one exchange.  For most companies we
  used the NYSE, but for Cisco, Intel, and Microsoft we used FINRA's
  Alternative Display Facility.



\item If there are several trades with the same time stamp, which is
  accurate up to seconds, then the median price across all such trades
  is taken to be the price at that time.
\item Discard a trade when the price is zero.
\item Discard a trade when the correction code is not zero.
\item Discard a trade when the condition code is a letter other than
  `E' or `F'.
\end{itemize}

\item[Synchronize Prices]: Regarding synchronization, prices of different assets
  are not updated at the same instant in time.  To make use of the statistical
  theory for constructing the realized measures one must decide how to ``align''
  prices in time so that they appear to be updated simultaneously.
  Barndorff-Nielsen et al. suggest constructing a set of refresh times
  $\{\tau_j\}_{j=1}^J$ which corresponds to a ``last most recently updated
  approach.''  The first refresh time $\tau_1$ is the first time at which all
  asset prices have been updated.  The subsequent refresh times are inductively
  defined so that $\tau_n$ is the first time at which all assets prices have
  been updated since $\tau_{n-1}$.  After cleansing and refreshing the data, one
  is left with the collection $\{X_{\tau_j}\}_{j=1}^J$ from which the realized
  kernels will be calculated.

\item[Jitter End Points]: For their asymptotic results to hold Barndorff-Nielsen
  et al. suggest jittering the first and last observations
  $\{X_{\tau_j}\}_{j=1}^J$.  We do this by taking the average of the first two
  observations and relabeling the resulting quantity as the first observation
  and taking the average of the last two observations and labeling the resulting
  quantity as the last observation.

\item[Calculate Bandwidths]:

  We follow \cite{bn2009rkip} when calculating each $H_i$ individually using the
  time series $\{X^{(i)}_{t_j}\}_{j=1}^n$ before it has been synchronized or
  jittered.  Fix $i$ and suppress it from the notation---we are only considering
  a single asset.  In particular, for asset $i$ the bandwidth $H$ is estimated
  as
\[
\hat H = c^* (\hat \xi^2)^{2/5} n^{3/5}
\]
where $c^* = 0.97$ for the Parzen kernel, $n$ is the number of observations, and
\[
\hat \xi^2 = \hat \omega^2 / \widehat {IV}.
\]
$\widehat {IV}$ is the realized variance sampled on a 20 minute grid.  $\hat
\omega^2$ is an estimate of the variance of $\{U_{t_i}\}_{i=1}^n$ and is given
by
\[
\hat \omega^2 = \frac{1}{q} \sum_{k=1}^q \hat \omega_k^2
\; \textmd{ with } \;
\hat \omega_k^2  = \frac{RV^{(k)}_{dense}}{2 n_{(k)}}.
\]
The quantity $RV^{(k)}_{dense}$ is the sum of square increments taken at a
high frequency.
\[
RV^{(k)}_{dense} = \sum_{j=0}^{n_k-1} {x_{j}^{(k)2}}, \; x_j^k = (X_{qj + k} -
X_{q(j-1) +k}), k = 1, \ldots, q.
\]
and $n_k$ is the number of observations elements in $\{x_j^k\}_{j=1}^{n_k}$.
For each time series we choose $q = \lfloor n / 195 \rfloor$, which is the
average number of ticks on that day per two minute period \citep{bn2009rkip}.

\end{description}

\bibliographystyle{abbrvnat}
\bibliography{rk,/Users/jwindle/Projects/RPackage/BayesLogit/Notes/bayeslogit}{}

\end{document}